\documentclass[12pt]{article}

\usepackage{amsmath,amssymb,amsfonts}
\usepackage{algorithm}
\usepackage{algorithmic}
\usepackage{graphicx}
\usepackage{textcomp}
\usepackage{bm}
\usepackage{array}
\usepackage{booktabs}
\usepackage{float}
\usepackage{placeins}
\usepackage{geometry}
\usepackage{hyperref}
\usepackage{adjustbox}

\newtheorem{lemma}{Lemma}
\newtheorem{theorem}{Theorem}
\newtheorem{corollary}{Corollary}

\newenvironment{proof}{\par\noindent\textit{Proof.}\ }{\hfill$\square$\par}

\newcommand{\EJ}{\mathrm{EJ}}
\newcommand{\Z}{\mathbb{Z}}

\newcommand{\om}{\omega}
\newcommand{\modEJ}{\operatorname{mod}_{\EJ_t}}
\newcommand{\dist}{d_{\EJ}}

\newcommand{\B}{\mathcal{B}_t}
\newcommand{\Kset}{\mathcal{K}_t}
\newcommand{\Count}{\operatorname{Count}_{\EJ}}

\title{Closed-Form and Constant-Time New-Source Selection for Fault-Tolerant Broadcasting in Dense Eisenstein--Jacobi Networks}

\author{Bader Albader\\
\small Department of Computer Science, Kuwait University, Kuwait\\
\small \texttt{albader@cs.ku.edu.kw}}

\date{}

\begin{document}

\maketitle

\begin{abstract}
Fault-tolerant broadcasting in dense Eisenstein--Jacobi networks requires an efficient recovery mechanism when faulty nodes interrupt the original broadcast structure. A recently published re-rooting-based broadcasting method for dense Eisenstein--Jacobi networks proves that, for any two faulty nodes, recovery can be performed by selecting a new source that is at maximum graph distance from both faults. However, the recovery step still benefits from a direct method that selects such a source without scanning the network or testing all boundary candidates. This paper develops a self-contained closed-form and constant-time new-source counting and selection method for dense Eisenstein--Jacobi networks. The two-fault problem is translated to an equivalent boundary-intersection problem involving the origin and a difference node. The distance-$t$ boundary, where $t$ is the network diameter, is partitioned into six directed sides of the Eisenstein--Jacobi hexagon. Because the network is a quotient network, the intersection equations must be solved modulo the defining Eisenstein--Jacobi lattice. Therefore, the proposed algorithms evaluate seven possible quotient-lattice shifts together with the $6\times 6$ side pairs, giving at most $7\cdot 6\cdot 6=252$ algebraic systems. For faults $0$ and $A$, the first algorithm counts all valid new sources exactly. For two arbitrary faults, the second algorithm selects one valid new source by solving translated side-pair systems, verifying the candidate, and shifting it back. Each system is either a nonparallel two-by-two linear system with at most one candidate, or a parallel system whose feasible candidates form an integer interval. Since the number of systems is fixed, both algorithms run in $O(1)$ time under the fixed-word arithmetic model. Computational validation over $500{,}000$ sampled fault pairs and $40{,}000$ re-rooting trials confirms that the direct selector always returns a valid new source and that the recovered broadcast reaches all non-faulty nodes in the tested settings.
\end{abstract}


\section{Introduction}
\label{sec:introduction}

Interconnection networks are widely used as models for communication in parallel and distributed systems. Their topology directly affects routing complexity, broadcasting time, diameter, scalability, and robustness under failures. Algebraic interconnection networks are especially useful because their nodes and links can be described by compact coordinate rules. Gaussian and Eisenstein--Jacobi networks are two important examples. Both are constructed from quotient rings, both admit modular coordinate descriptions, and both have small diameter relative to the number of nodes \cite{flahive2010topology,martinez2006dense,beivide2005gaussian}. Recent work on circulant and algebraic network topologies for networks-on-chip continues to emphasize coordinate-based routing, small diameter, deadlock-free communication, and resilience under faults \cite{monakhov2021adaptive,romanov2023ringsplit,monakhova2023relative,sukhov2024virtual,samala2025ftam,yu2025faultqos}.

Dense Eisenstein--Jacobi networks form a degree-six family of hexagonal networks. In this paper, the network is denoted by $\EJ_t$, where $t$ is the network diameter. The dense network is generated by
\begin{equation}
\alpha=(t+1)+t\om,
\label{eq:alpha}
\end{equation}
where $\om$ is an Eisenstein--Jacobi unit satisfying $\om^2+\om+1=0$. The number of nodes is
\begin{equation}
N=3t^2+3t+1.
\label{eq:N}
\end{equation}
Equivalently, if the original network parameter is written as $n=t+1$, then $N=3n^2-3n+1$. This paper uses $t$ throughout because it is the diameter notation used in the Eisenstein--Jacobi re-rooting framework.

In a fault-free dense Eisenstein--Jacobi network, a source can broadcast to all nodes within $t$ steps. When one or more nodes fail, part of the original broadcast structure may become unusable. The re-rooting method addresses this problem by choosing a new source that is farthest from the faulty nodes. For two faulty nodes $A$ and $B$, the target condition is
\begin{equation}
\dist(NS,A)=t,
\qquad
\dist(NS,B)=t.
\label{eq:valid-source-intro}
\end{equation}
The published Eisenstein--Jacobi re-rooting paper proves the existence of such a node for two faulty nodes and uses this fact to recover the broadcast \cite{albader2026ejrerooting}. The remaining algorithmic question is how to select such a source as directly as possible.

This paper answers that question in a self-contained way for dense Eisenstein--Jacobi networks. The key observation is that the distance-$t$ boundary is a six-sided hexagon. Therefore, the valid-source problem can be reduced to a fixed set of side-pair intersections. The construction developed here is native to the Eisenstein--Jacobi geometry: it uses six boundary sides, $252$ translated side-pair systems, and the Eisenstein--Jacobi graph distance in \eqref{eq:ej-distance}.

Given two faulty nodes $A$ and $B$, define the difference node
\begin{equation}
C=\modEJ(B-A).
\label{eq:difference-node-intro}
\end{equation}
By translation invariance, it is enough to find a point $P$ satisfying
\begin{equation}
\dist(P,0)=t,
\qquad
\dist(P,C)=t.
\label{eq:P-condition-intro}
\end{equation}
The new source is then obtained by shifting back:
\begin{equation}
NS=\modEJ(A+P).
\label{eq:shift-back-intro}
\end{equation}
Thus, the selection problem becomes an intersection problem between two distance-$t$ boundaries.

The main contributions of this paper are as follows.
\begin{itemize}
\item The distance-$t$ boundary of $\EJ_t$ is partitioned into six directed sides, each containing exactly $t$ nodes.
\item For faults $0$ and $A$, the exact number of valid new sources is expressed as a closed finite sum of translated side-pair intersection counts over seven quotient shifts and thirty-six side pairs.
\item A constant-time counting algorithm is developed for faults $0$ and $A$, and a direct side-pair selector is developed for two arbitrary faulty nodes. Both use the same fixed set of $252$ translated side-pair systems.
\item The correctness of the selector is proved through soundness and completeness theorems.
\item The time complexity of both the counting algorithm and the selection algorithm is shown to be $O(1)$ because each checks at most $252$ translated side pairs, independent of $t$, $N$, or the boundary size $6t$.
\item Worked examples and figures are included to illustrate the counting and selection procedures.
\end{itemize}

The rest of the paper is organized as follows. Section~\ref{sec:related} reviews related work. Section~\ref{sec:preliminaries} gives the network model and distance notation. Section~\ref{sec:boundary} defines the six-side boundary partition. Section~\ref{sec:counting} gives the exact number of valid new sources for faults $0$ and $A$. Section~\ref{sec:direct-selection} presents the constant-time selector for arbitrary fault pairs. Section~\ref{sec:complexity} compares the complexity with search-based alternatives. Section~\ref{sec:validation} presents the computational validation results. Section~\ref{sec:discussion} discusses the role of the result within the re-rooting framework, and Section~\ref{sec:conclusion} concludes the paper.

\section{Related Work}
\label{sec:related}

Gaussian and Eisenstein--Jacobi interconnection networks were introduced and analyzed as quotient-ring networks with strong symmetry and compact algebraic descriptions \cite{flahive2010topology}. Dense Gaussian networks and Gaussian coordinate representations were later studied for multiprocessor and network-on-chip settings, where small diameter and regular degree are important design properties \cite{martinez2006dense,beivide2005gaussian}. Hierarchical extensions and related circulant constructions show that algebraic network families remain useful for large-scale interconnection design \cite{vallejo2008hierarchical}.

Recent research on circulant and algebraic interconnection networks shows that compact coordinate representations remain useful for scalable routing. Monakhov \emph{et al.} studied adaptive shortest-path search for networks-on-chip based on circulant topologies, while Romanov \emph{et al.} investigated routing algorithms for two-dimensional optimal circulant NoCs \cite{monakhov2021adaptive,romanov2020circulant}. Monakhova \emph{et al.} later developed relative-addressing routing algorithms for optimal degree-four circulant networks, and Sukhov \emph{et al.} proposed a virtual coordinate system for circulant NoC routing \cite{monakhova2023relative,sukhov2024virtual}.

Fault tolerance is also a central concern in networks-on-chip and related interconnection systems. Surveys and monographs on on-chip communication describe how node and link failures can affect routing reliability and system-level performance \cite{kliazovich2010survey,pasricha2008onchip,flich2010designing}. Recent NoC studies continue to investigate fault-tolerant and adaptive mechanisms, including permanent-fault router architectures, deadlock-free routing for circulant NoCs, fault-tolerant application mapping, and quality-of-service-aware routing methods \cite{rashid2020router,romanov2023ringsplit,samala2025ftam,yu2025faultqos}.

Broadcasting and spanning-tree construction are closely related to the present work because a broadcast algorithm must deliver information from one source to all non-faulty nodes while avoiding failed components. Independent spanning-tree methods and classical studies of highly regular topologies illustrate the importance of having alternative communication structures under faults \cite{wu1999distributed,saad1988topological}. The present work is most closely related to the published re-rooting-based fault-tolerant broadcasting method for dense Eisenstein--Jacobi networks \cite{albader2026ejrerooting}. That paper establishes the two-fault existence result and uses it to restore broadcast coverage by selecting a new source at distance $t$ from both faulty nodes. The present paper does not replace that broadcasting model. Instead, it strengthens the source-selection step by replacing boundary search with closed-form counting and constant-time algebraic selection.

\section{Preliminaries and Network Model}
\label{sec:preliminaries}

\subsection{Eisenstein--Jacobi Coordinates}
A node in $\EJ_t$ is represented by an Eisenstein--Jacobi coordinate
\begin{equation}
X=x+y\om,
\end{equation}
which we also write as the ordered pair $X=(x,y)$. Since $\EJ_t$ is a quotient network, different coordinate pairs may represent the same node. The notation
\begin{equation}
\modEJ(X)
\end{equation}
denotes the canonical representative of $X$ in the selected coordinate region of $\EJ_t$.

The generator in \eqref{eq:alpha} gives a network with
\begin{equation}
N=3t^2+3t+1
\end{equation}
nodes and diameter $t$. Throughout the paper, all coordinate sums and differences that represent network nodes are reduced by $\modEJ(\cdot)$ when necessary.

\subsection{Hexagonal Graph Distance}
For a reduced coordinate $X=(x,y)$, the Eisenstein--Jacobi graph distance from the origin is
\begin{equation}
\dist(0,X)=\max\{|x|,|y|,|x+y|\}.
\label{eq:ej-distance}
\end{equation}
For two nodes $U,V\in\EJ_t$, the distance is computed by reducing the difference:
\begin{equation}
\dist(U,V)=\dist(0,\modEJ(U-V)).
\label{eq:distance-difference}
\end{equation}
The distance is translation invariant. Hence, for any nodes $U,V,T\in\EJ_t$,
\begin{equation}
\dist(U,V)=\dist(\modEJ(U+T),\modEJ(V+T)).
\label{eq:translation-invariance}
\end{equation}
This property is the reason why an arbitrary two-fault problem can be translated to the origin and a difference node.

\subsection{Boundary Nodes and Valid New Sources}
The distance-$t$ boundary around the origin is
\begin{equation}
\B=\{P\in\EJ_t:\dist(P,0)=t\}.
\label{eq:boundary}
\end{equation}
More generally, the distance-$t$ boundary around a node $A$ is
\begin{equation}
A+\B=\{\modEJ(A+P):P\in\B\}.
\end{equation}
A node $P$ is a valid new source with respect to faults $0$ and $A$ if
\begin{equation}
P\in \B\cap(A+\B).
\end{equation}
The number of such nodes is denoted by
\begin{equation}
\Count(A,t)=|\B\cap(A+\B)|.
\label{eq:count-definition}
\end{equation}
For arbitrary faults $A$ and $B$, a node $NS$ is valid if it satisfies \eqref{eq:valid-source-intro}.

\subsection{Integer Labeling}
The coordinate representation is compatible with an integer label modulo $N$. If the original dense Eisenstein--Jacobi parameter is $n=t+1$, then the integer label of $x+y\om$ is
\begin{equation}
\phi(x+y\om)\equiv t x-(t+1)y \pmod{N},
\label{eq:integer-label}
\end{equation}
where $N=3t^2+3t+1$. This is the same as
\begin{equation}
\phi(x+y\om)\equiv (n-1)x-ny \pmod{3n^2-3n+1}.
\end{equation}
The labeling is useful in implementation because differences can be computed either in coordinate form or through their corresponding modular labels. The mathematical development below uses coordinates because the boundary geometry is clearer in the hexagonal representation.

\begin{figure}[H]
\centering
\includegraphics[width=0.6\linewidth]{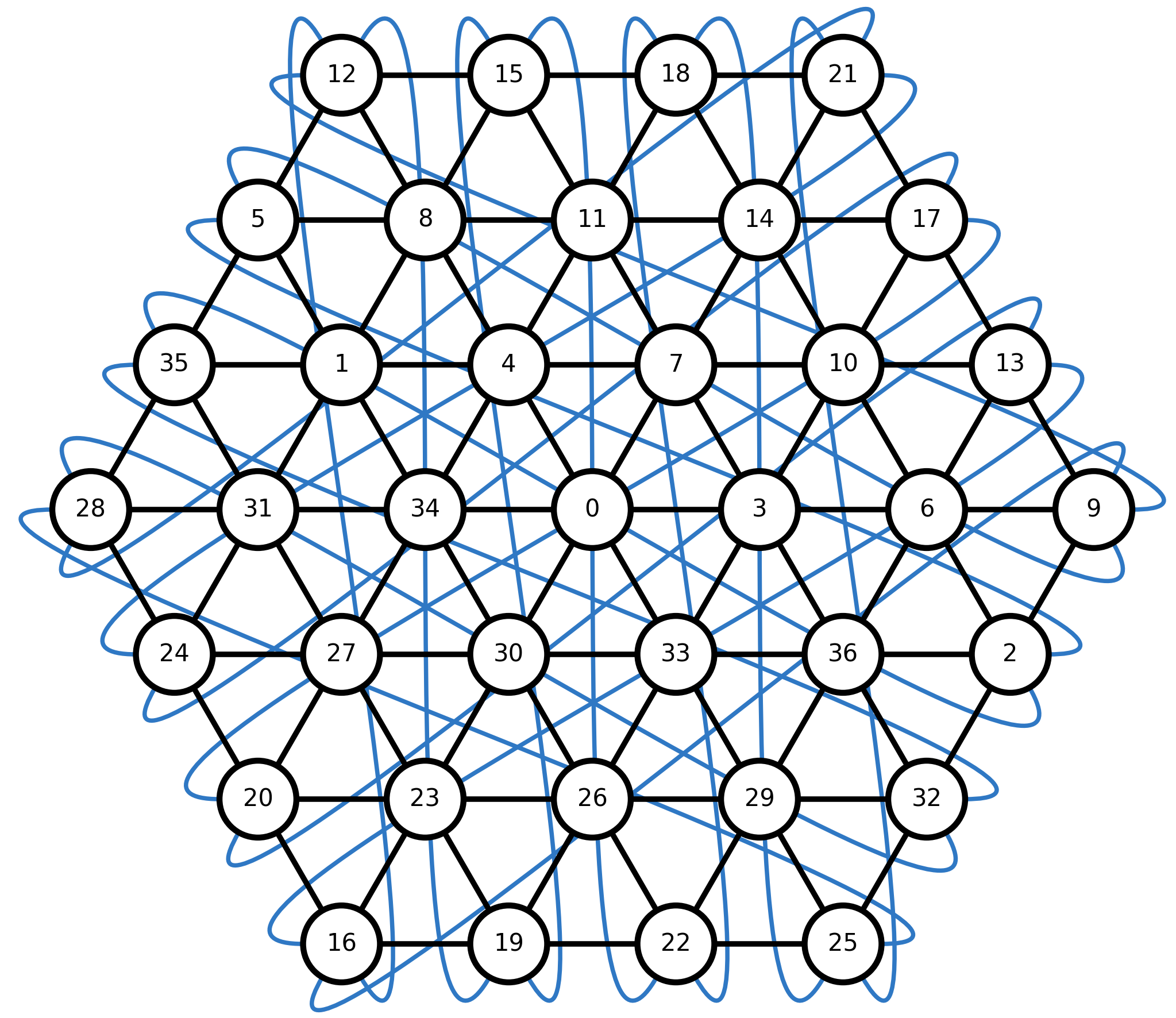}
\caption{Example dense Eisenstein--Jacobi network $H_3$ with integer node labels. This topology illustrates the hexagonal quotient structure and the wrap-around nature used by the counting and selection algorithms.}
\label{fig:ej-h3-int}
\end{figure}

\section{Six-Side Boundary Partition}
\label{sec:boundary}

The boundary $\B$ is a hexagon. To count and select valid new sources without scanning $\B$, we partition it into six directed sides. Let
\begin{align}
V_1&=(t,0),       & u_1&=(-1,1),\nonumber\\
V_2&=(0,t),       & u_2&=(-1,0),\nonumber\\
V_3&=(-t,t),      & u_3&=(0,-1),\nonumber\\
V_4&=(-t,0),      & u_4&=(1,-1),\nonumber\\
V_5&=(0,-t),      & u_5&=(1,0),\nonumber\\
V_6&=(t,-t),      & u_6&=(0,1).
\label{eq:vertices-directions}
\end{align}
For $i=1,\ldots,6$, define
\begin{equation}
S_i=\{V_i+s u_i:0\leq s\leq t-1,\ s\in\Z\}.
\label{eq:side-definition}
\end{equation}
Each side has exactly $t$ nodes. The endpoint convention in \eqref{eq:side-definition} includes the first vertex of each side and excludes the next vertex. Therefore the six sets are disjoint.

\begin{figure}[H]
\centering
\includegraphics[width=0.6\linewidth]{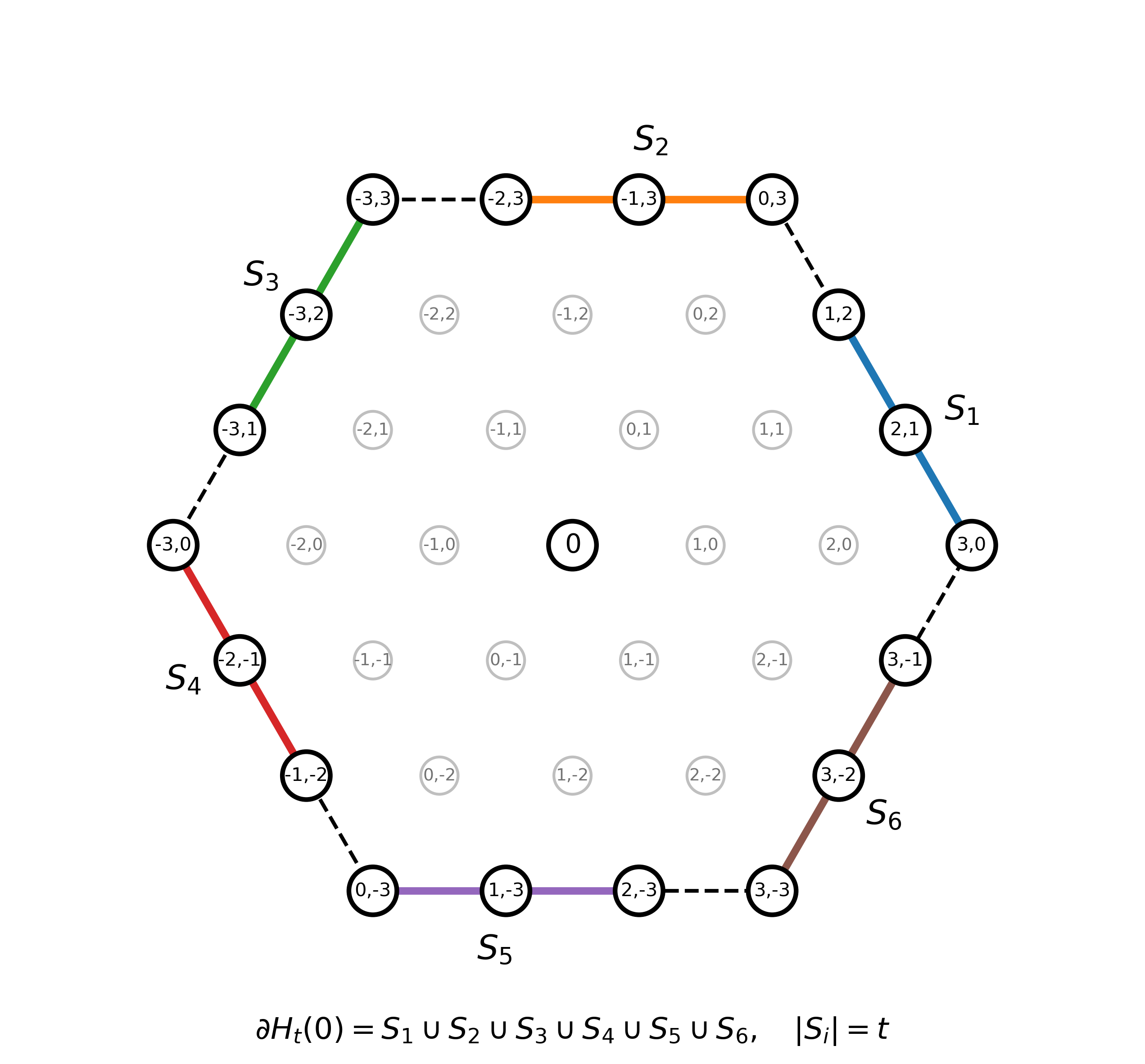}
\caption{Six-side partition of the Eisenstein--Jacobi distance-$t$ boundary. The boundary is written as $\mathcal{B}_t=S_1\cup\cdots\cup S_6$, where each side has $t$ boundary nodes.}
\label{fig:six-side-boundary}
\end{figure}

\begin{lemma}[Boundary partition]
The sets $S_1,\ldots,S_6$ form a disjoint partition of the distance-$t$ boundary:
\begin{equation}
\B=\bigcup_{i=1}^{6}S_i,
\qquad
S_i\cap S_j=\emptyset \quad (i\neq j).
\end{equation}
Consequently, $|\B|=6t$.
\end{lemma}

\begin{proof}
For every point $P=(x,y)$ in one of the six sets, the value of $\max\{|x|,|y|,|x+y|\}$ is exactly $t$. For example, on $S_1$ we have $P=(t-s,s)$, so $x+y=t$ and $0\leq x,y\leq t$. Hence $\dist(P,0)=t$. The other five sides are analogous and correspond to $y=t$, $x=-t$, $x+y=-t$, $y=-t$, and $x=t$. Conversely, if $\dist(P,0)=t$, then one of the six quantities $x$, $y$, $-x$, $-y$, $x+y$, or $-(x+y)$ is equal to $t$. This places $P$ on one of the six sides. The half-open endpoint convention assigns each vertex to exactly one side, so the sides are disjoint. Since each side contains $t$ integer points, the boundary contains $6t$ nodes.
\end{proof}

\section{Number of Valid New Sources for Faults $0$ and $A$}
\label{sec:counting}

This section gives the missing counting component: for a given node $A$, how many valid new sources exist when the faulty nodes are $0$ and $A$? Figure~\ref{fig:quotient-shift} illustrates why quotient-lattice shifts must be included, and Figure~\ref{fig:count-intersection} shows a counting example. The answer is the size of the quotient boundary intersection
\begin{equation}
\B\cap(A+\B) \quad \text{in } \EJ_t.
\end{equation}
A direct planar intersection of the six sides is not sufficient, because two coordinate points may represent the same node modulo the Eisenstein--Jacobi lattice. Therefore the correct side-pair equation must include the finite set of quotient-lattice shifts that can occur between radius-$t$ representatives.

\subsection{Quotient-Lattice Shifts}
Let
\begin{equation}
\Kset=\{(0,0),\pm L_1,\pm L_2,\pm L_3\},
\label{eq:kernel-shift-set}
\end{equation}
where
\begin{equation}
L_1=(t+1,t),\quad
L_2=(2t+1,-t-1),\quad
L_3=(t,-2t-1).
\end{equation}
These are the zero shift and the six shortest nonzero lattice shifts induced by the labeling
\begin{equation}
\phi(x,y)\equiv tx-(t+1)y \pmod{N}.
\end{equation}
They satisfy $\phi(L)=0$ for every $L\in\Kset$, and hence adding any such shift does not change the represented node in $\EJ_t$.

\begin{lemma}[Finite shift sufficiency]
\label{lem:finite-shifts}
If $P,Q,A$ are canonical representatives in the radius-$t$ hexagon and
\begin{equation}
P\equiv A+Q \pmod{\alpha},
\end{equation}
with $P,Q\in\B$, then there exists a shift $L\in\Kset$ such that
\begin{equation}
P=A+Q+L.
\end{equation}
\end{lemma}

\begin{proof}
The equality in the quotient means that $P-A-Q$ is a lattice vector in the kernel of the labeling map. Since $P,A,Q$ are canonical radius-$t$ representatives,
\begin{equation}
D(P-A-Q)\leq D(P)+D(A)+D(Q)\leq 3t,
\end{equation}
where $D(x,y)=\max\{|x|,|y|,|x+y|\}$. Therefore, only kernel vectors whose hexagonal length is at most $3t$ can appear in a boundary-intersection equation.

It remains to identify these short kernel vectors. The kernel lattice is generated by two independent wrap-around vectors, for example
\begin{equation}
L_1=(t+1,t),\qquad L_3=(t,-2t-1),
\end{equation}
and the third shortest direction is $L_2=L_1+L_3=(2t+1,-t-1)$. Hence every kernel vector can be written as
\begin{equation}
K(a,b)=aL_1+bL_3,\qquad a,b\in\mathbb{Z}.
\end{equation}
A direct substitution gives
\begin{equation}
K(a,b)=\big(a(t+1)+bt,\; at-b(2t+1)\big).
\end{equation}
Checking the integer pairs $(a,b)$ shows that the only nonzero vectors with $D(K(a,b))\leq 3t$ are obtained from
\begin{equation}
(a,b)\in\{(1,0),(-1,0),(0,1),(0,-1),(1,1),(-1,-1)\},
\end{equation}
which are precisely $\pm L_1$, $\pm L_3$, and $\pm L_2$. All other integer combinations have at least one of $|x|$, $|y|$, or $|x+y|$ greater than $3t$. Thus the only possible shifts in the present boundary-intersection problem are the zero vector and the six shifts in \eqref{eq:kernel-shift-set}. Hence $P-A-Q=L$ for some $L\in\Kset$, which gives $P=A+Q+L$.
\end{proof}

\subsection{Translated Side-Pair Intersection Formula}
Let
\begin{equation}
A=(a_1,a_2).
\end{equation}
For $L\in\Kset$ and $i,j\in\{1,\ldots,6\}$, define
\begin{equation}
\eta_{Lij}(A,t)=|S_i\cap(A+S_j+L)|.
\label{eq:eta-definition}
\end{equation}
Using \eqref{eq:side-definition}, a point lies in $S_i\cap(A+S_j+L)$ if and only if there exist integers $s,u$ with $0\leq s,u\leq t-1$ such that
\begin{equation}
V_i+s u_i=A+V_j+u u_j+L.
\label{eq:side-pair-equation}
\end{equation}
Equivalently,
\begin{equation}
\begin{bmatrix}
u_i & -u_j
\end{bmatrix}
\begin{bmatrix}s\\u\end{bmatrix}
=A+V_j+L-V_i,
\label{eq:linear-system-compact}
\end{equation}
where the two columns are the two-dimensional vectors $u_i$ and $-u_j$.

Let
\begin{equation}
M_{ij}=\begin{bmatrix}u_i&-u_j\end{bmatrix},
\qquad
b_{Lij}=A+V_j+L-V_i.
\end{equation}
Then \eqref{eq:linear-system-compact} becomes
\begin{equation}
M_{ij}\begin{bmatrix}s\\u\end{bmatrix}=b_{Lij}.
\label{eq:side-pair-linear}
\end{equation}
The contribution $\eta_{Lij}(A,t)$ is the number of integer solutions of \eqref{eq:side-pair-linear} satisfying $0\leq s,u\leq t-1$.

\begin{theorem}[Exact count for faults $0$ and $A$]
\label{thm:exact-count}
For any node $A\in\EJ_t$, the exact number of valid new sources for faulty nodes $0$ and $A$ is
\begin{equation}
\boxed{
\Count(A,t)=\sum_{L\in\Kset}\sum_{i=1}^{6}\sum_{j=1}^{6}\eta_{Lij}(A,t)
}
\label{eq:ej-count-formula}
\end{equation}
where $\eta_{Lij}(A,t)$ is obtained from the translated side-pair system \eqref{eq:side-pair-linear}.
\end{theorem}

\begin{proof}
A node $P$ is valid for faults $0$ and $A$ exactly when
\begin{equation}
\dist(P,0)=t
\quad\text{and}\quad
\dist(P,A)=t.
\end{equation}
The first condition means $P\in\B$. By translation invariance, the second condition means that $P-A$ represents a boundary node. Thus there exists $Q\in\B$ such that
\begin{equation}
P\equiv A+Q \pmod{\alpha}.
\end{equation}
By Lemma~\ref{lem:finite-shifts}, this quotient equality is equivalent, for boundary representatives, to
\begin{equation}
P=A+Q+L
\end{equation}
for some $L\in\Kset$. Using the disjoint boundary partition, $P\in S_i$ and $Q\in S_j$ for unique side indices $i$ and $j$. Therefore every valid source is counted by exactly one translated side-pair term $\eta_{Lij}(A,t)$, with the half-open endpoint convention preventing double counting at vertices. Conversely, every solution of a translated side-pair system gives a point $P\in\B$ satisfying $P\equiv A+Q$ with $Q\in\B$, and hence $\dist(P,A)=t$. Summing over the seven shifts and the thirty-six side pairs proves \eqref{eq:ej-count-formula}.
\end{proof}

\begin{figure}[H]
\centering
\includegraphics[width=0.6\linewidth]{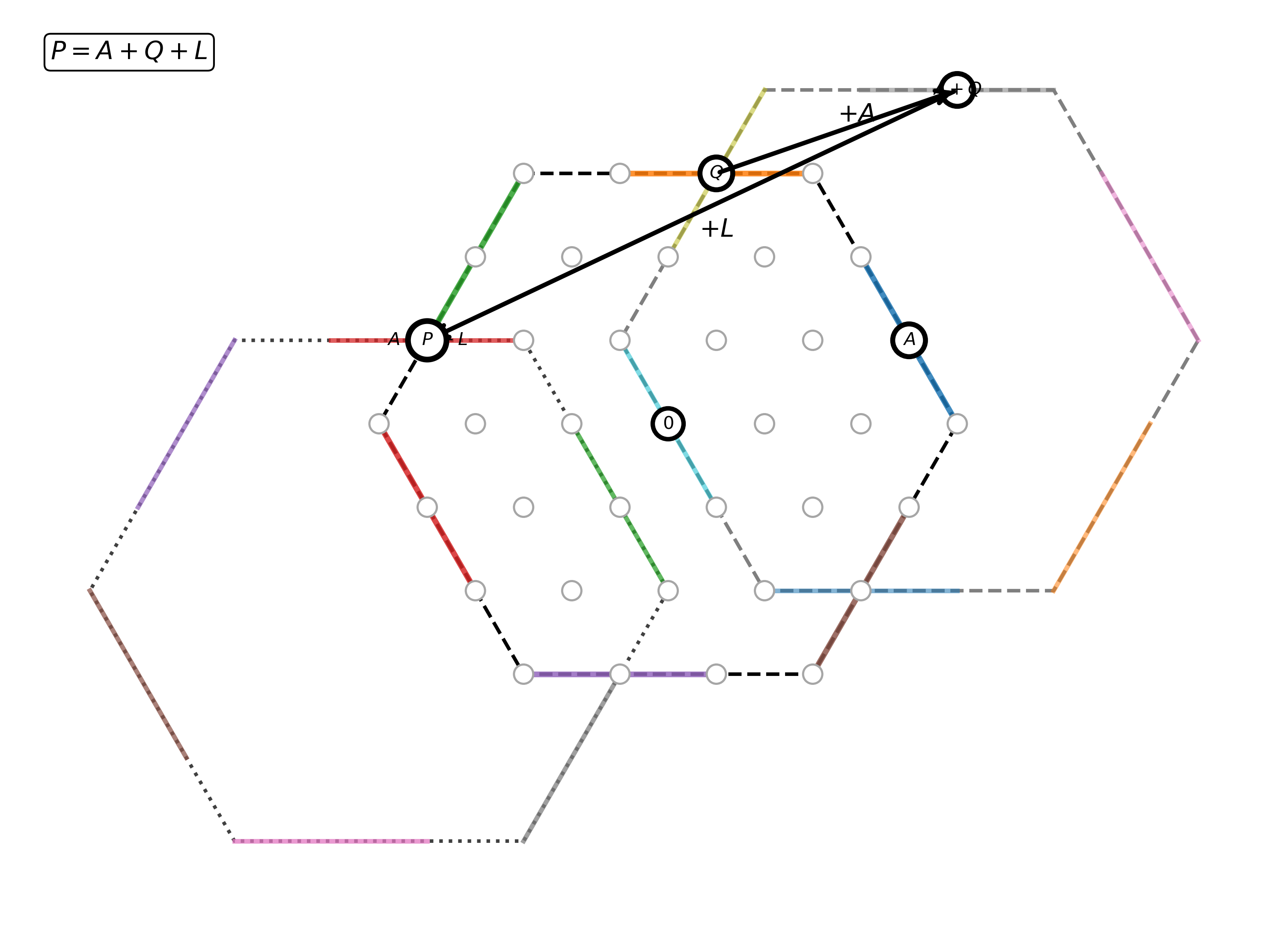}
\caption{Illustration of the quotient-lattice correction. A valid boundary intersection in the finite dense Eisenstein--Jacobi network satisfies $P\equiv A+Q\pmod{\alpha}$, which in canonical coordinates becomes $P=A+Q+L$ for a suitable lattice shift $L$. Without the shift term $L$, a planar side-pair intersection may miss wrap-around boundary intersections.}
\label{fig:quotient-shift}
\end{figure}

The count in \eqref{eq:ej-count-formula} should be interpreted as a closed finite evaluation rather than as a boundary scan. A boundary scan tests the $6t$ possible values of $P$ one by one. In contrast, the translated side-pair formula groups boundary nodes into six parametric sides and solves the corresponding intersection equations algebraically. Parallel side pairs may represent several nodes at once, and nonparallel side pairs contribute at most one node. This is why the count can be evaluated with a fixed number of arithmetic tests even though the number of boundary nodes grows linearly with $t$.

\subsection{Closed Evaluation of Each Side Pair}
The formula \eqref{eq:ej-count-formula} is constant-size because it contains at most $7\cdot 6\cdot 6=252$ translated side-pair terms. Each term is evaluated as follows.

If $\det(M_{ij})\neq 0$, then \eqref{eq:side-pair-linear} has at most one solution. Let
\begin{equation}
\begin{bmatrix}s^*\\u^*\end{bmatrix}=M_{ij}^{-1}b_{Lij}.
\end{equation}
Then
\begin{equation}
\eta_{Lij}(A,t)=
\begin{cases}
1, & s^*,u^*\in\Z,\ 0\leq s^*,u^*\leq t-1,\\
0, & \text{otherwise.}
\end{cases}
\label{eq:eta-nonparallel}
\end{equation}
If $\det(M_{ij})=0$, then the two sides are parallel. In this case, either the equations are inconsistent and $\eta_{Lij}(A,t)=0$, or they reduce to one linear equation in two bounded integer parameters. The feasible set is an integer interval, and its size is
\begin{equation}
\eta_{Lij}(A,t)=\max\{0,U-L+1\},
\label{eq:eta-parallel}
\end{equation}
where $L$ and $U$ are the lower and upper bounds obtained from $0\leq s,u\leq t-1$ and the remaining linear equation.

\subsection{Algorithmic Evaluation of the Count}
Algorithm~\ref{alg:count-ej} evaluates \eqref{eq:ej-count-formula} directly. It is included separately from the source-selection algorithm because it answers a different question: when the faulty nodes are $0$ and $A$, it returns the number of valid new-source choices, not just one choice.

\begin{algorithm}[H]
\caption{Constant-Time Count of Valid New Sources for Faults $0$ and $A$}
\label{alg:count-ej}
\begin{algorithmic}[1]
\REQUIRE Network diameter $t$; node $A\in\EJ_t$
\ENSURE The number $\Count(A,t)$ of valid new sources for faults $0$ and $A$
\STATE $Q\leftarrow 0$
\STATE Construct $\Kset=\{(0,0),\pm(t+1,t),\pm(2t+1,-t-1),\pm(t,-2t-1)\}$
\FORALL{$L\in\Kset$}
    \FOR{$i=1$ to $6$}
        \FOR{$j=1$ to $6$}
            \STATE $M\leftarrow [u_i\; -u_j]$
            \STATE $b\leftarrow A+V_j+L-V_i$
            \IF{$\det(M)\neq 0$}
                \STATE Solve $M[s\;u]^T=b$ by Cramer's rule
                \IF{$s,u$ are integers and $0\leq s,u\leq t-1$}
                    \STATE $Q\leftarrow Q+1$
                \ENDIF
            \ELSE
                \STATE Check consistency of the parallel system
                \IF{the parallel system is consistent}
                    \STATE Compute the feasible interval $[L_b,U_b]$
                    \IF{$L_b\leq U_b$}
                        \STATE $Q\leftarrow Q+(U_b-L_b+1)$
                    \ENDIF
                \ENDIF
            \ENDIF
        \ENDFOR
    \ENDFOR
\ENDFOR
\RETURN $Q$
\end{algorithmic}
\end{algorithm}

\begin{theorem}[Correctness and complexity of the counting algorithm]
Algorithm~\ref{alg:count-ej} returns the exact value of $\Count(A,t)$ and runs in $O(1)$ time with respect to $t$ and $N$ under the fixed-word arithmetic model.
\end{theorem}

\begin{proof}
For each shift $L\in\Kset$ and each pair $(i,j)$, the algorithm solves exactly the translated side-pair system
\begin{equation}
M_{ij}\begin{bmatrix}s\\u\end{bmatrix}=A+V_j+L-V_i
\end{equation}
with the restrictions $0\leq s,u\leq t-1$. If $\det(M_{ij})\neq 0$, the two sides are nonparallel and there is at most one intersection point; the algorithm adds one precisely when that unique solution is integral and lies in the two parameter ranges. If $\det(M_{ij})=0$, the two sides are parallel; the algorithm first tests consistency and then counts the feasible integer interval by $U_b-L_b+1$. Therefore the value added for $(L,i,j)$ is exactly $\eta_{Lij}(A,t)$.

The algorithm sums these values over all seven shifts and all thirty-six side pairs. By Theorem~\ref{thm:exact-count}, the sum is exactly $\Count(A,t)$. The running time is constant because the loops range over the fixed set $\Kset\times\{1,\ldots,6\}\times\{1,\ldots,6\}$, which has $7\cdot6\cdot6=252$ elements, and each translated side-pair evaluation uses only a fixed number of arithmetic, range-checking, and interval-bound operations. No node scan and no boundary scan is performed. Hence the complexity is $O(1)$ under the fixed-word arithmetic model.
\end{proof}

\begin{figure}[H]
\centering
\includegraphics[width=0.6\linewidth]{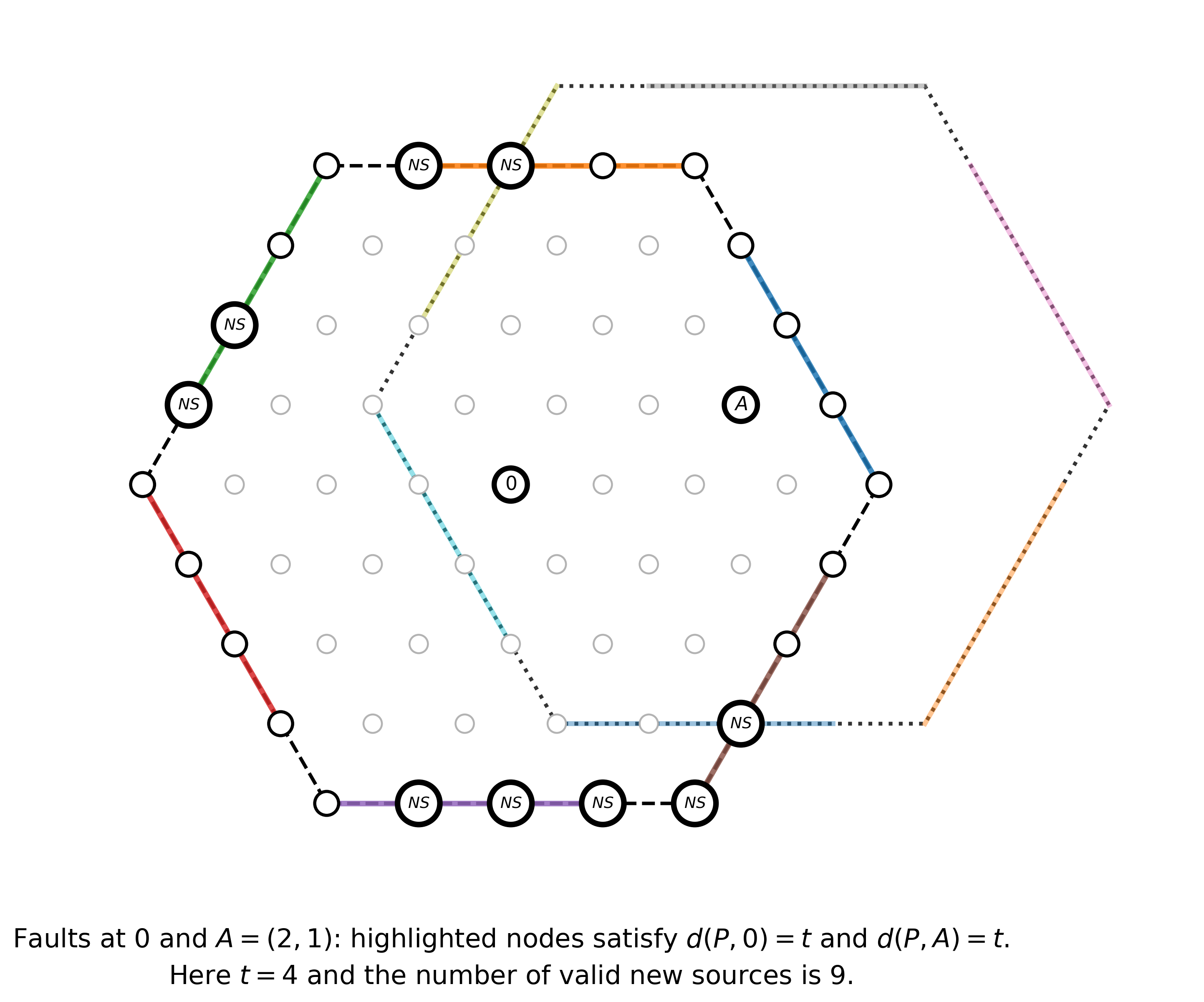}
\caption{Counting valid new sources for faults $0$ and $A$. The highlighted boundary nodes are exactly those $P$ satisfying $\dist(P,0)=t$ and $\dist(P,A)=t$. In the example shown, $t=4$, $A=(2,1)$, and the number of valid new sources is six.}
\label{fig:count-intersection}
\end{figure}

\begin{corollary}[Existence from the count]
If the re-rooting existence theorem guarantees a common distance-$t$ node for faults $0$ and $A$, then
\begin{equation}
\Count(A,t)>0.
\end{equation}
Moreover, \eqref{eq:ej-count-formula} gives the exact number of available choices.
\end{corollary}

\begin{proof}
The re-rooting existence theorem gives at least one node in $\B\cap(A+\B)$. Theorem~\ref{thm:exact-count} counts that intersection exactly, so the count is positive and equals the number of choices.
\end{proof}

\subsection{Worked Counting Examples for $t=3$}
For $t=3$, the network contains
\begin{equation}
N=3(3)^2+3(3)+1=37
\end{equation}
nodes and the boundary contains $6t=18$ nodes.

\textbf{Example 1: faults $0$ and $A=(1,0)$.}
A direct evaluation of \eqref{eq:ej-count-formula} gives
\begin{equation}
\Count((1,0),3)=13.
\end{equation}
The thirteen valid choices are
\begin{equation}
\begin{gathered}
(0,3),\ (-1,3),\ (-2,3),\ (-3,3),\ (-3,2),\ (-3,1),\ (-3,0),\\
(-2,-1),\ (-1,-2),\ (0,-3),\ (1,-3),\ (2,-3),\ (3,-3).
\end{gathered}
\end{equation}
Each of these points is at distance $3$ from both $0$ and $(1,0)$.

\textbf{Example 2: faults $0$ and $A=(2,0)$.}
The count is
\begin{equation}
\Count((2,0),3)=8.
\end{equation}
The valid choices are
\begin{equation}
\begin{gathered}
(0,3),\ (-1,3),\ (-2,3),\ (-3,1),\ (-3,0),\ (1,-3),\\
 (2,-3),\ (3,-3).
\end{gathered}
\end{equation}

\textbf{Example 3: faults $0$ and $A=(1,1)$.}
The count is
\begin{equation}
\Count((1,1),3)=9,
\end{equation}
with valid choices
\begin{equation}
\begin{gathered}
(-2,3),\ (-3,3),\ (-3,2),\ (-3,1),\ (0,-3),\\
 (1,-3),\ (2,-3),\ (3,-3),\ (3,-2).
\end{gathered}
\end{equation}
These examples show why the counting section is useful: the valid new source is not only guaranteed to exist, but the exact number of available candidates can also be computed without scanning the boundary.

\section{Direct Constant-Time New-Source Selection}
\label{sec:direct-selection}

This section presents the main selection algorithm. It selects one valid new source for two arbitrary faulty nodes by translating the problem, solving a fixed number of Eisenstein--Jacobi side-pair systems, verifying the candidate, and shifting back.

\subsection{Translation to the Origin}
Let $A,B\in\EJ_t$ be two faulty nodes. Define
\begin{equation}
C=\modEJ(B-A).
\end{equation}
If $C=(c_1,c_2)$, the translated problem is to find $P$ such that
\begin{equation}
\dist(P,0)=t,
\qquad
\dist(P,C)=t.
\end{equation}
Once $P$ is found, set
\begin{equation}
NS=\modEJ(A+P).
\end{equation}

\begin{lemma}[Translation correctness]
If $P$ satisfies $\dist(P,0)=t$ and $\dist(P,C)=t$, where $C=\modEJ(B-A)$, then $NS=\modEJ(A+P)$ satisfies
\begin{equation}
\dist(NS,A)=t,
\qquad
\dist(NS,B)=t.
\end{equation}
\end{lemma}

\begin{proof}
Using translation invariance,
\begin{equation}
\dist(NS,A)=\dist(\modEJ(A+P),A)=\dist(P,0)=t.
\end{equation}
Since $C=\modEJ(B-A)$, we have $B=\modEJ(A+C)$. Therefore
\begin{equation}
\begin{gathered}
\dist(NS,B)=\dist(\modEJ(A+P),\\
\modEJ(A+C))=\dist(P,C)=t.
\end{gathered}
\end{equation}
Thus $NS$ is at distance $t$ from both faults.
\end{proof}

\subsection{Side-Pair Algebra}
The point $P$ must lie in
\begin{equation}
\B\cap(C+\B).
\end{equation}
Therefore, for some shift $L\in\Kset$ and some side pair $(i,j)$,
\begin{equation}
P\in S_i\cap(C+S_j+L).
\end{equation}
Using the same side parameterization, this means
\begin{equation}
V_i+s u_i=C+V_j+u u_j+L,
\label{eq:selector-side-equation}
\end{equation}
with integers $0\leq s,u\leq t-1$. This is again the two-by-two system
\begin{equation}
M_{ij}\begin{bmatrix}s\\u\end{bmatrix}=C+V_j+L-V_i.
\label{eq:selector-linear-system}
\end{equation}
The algorithm checks at most $7\cdot6\cdot6=252$ translated systems. A nonparallel system gives one possible candidate. A parallel system gives either no candidate or an interval of candidates. In either case, only constant-time arithmetic is required.

\subsection{Algorithm 2: Direct Side-Pair Selector}
Algorithm~\ref{alg:direct-ej} is the proposed constant-time selector. It checks a fixed set of translated hexagonal side pairs and therefore does not enumerate the full node set or the distance-$t$ boundary.

\begin{algorithm}[H]
\caption{Constant-Time Direct New-Source Selection in $\EJ_t$}
\label{alg:direct-ej}
\begin{algorithmic}[1]
\REQUIRE Network diameter $t$; faulty nodes $A,B\in\EJ_t$
\ENSURE A valid new source $NS$
\STATE $C\leftarrow \modEJ(B-A)$
\STATE Construct $\Kset=\{(0,0),\pm(t+1,t),\pm(2t+1,-t-1),\pm(t,-2t-1)\}$
\FORALL{$L\in\Kset$}
    \FOR{$i=1$ to $6$}
        \FOR{$j=1$ to $6$}
            \STATE $M\leftarrow [u_i\; -u_j]$
            \STATE $b\leftarrow C+V_j+L-V_i$
            \IF{$\det(M)\neq 0$}
                \STATE Solve $M[s\;u]^T=b$ by Cramer's rule
                \IF{$s,u$ are integers and $0\leq s,u\leq t-1$}
                    \STATE $P\leftarrow V_i+s u_i$
                    \IF{$\dist(P,0)=t$ and $\dist(P,C)=t$}
                        \RETURN $\modEJ(A+P)$
                    \ENDIF
                \ENDIF
            \ELSE
                \STATE Compute the feasible interval $[L_b,U_b]$ for the parallel system
                \IF{$L_b\leq U_b$}
                    \STATE Choose $s\leftarrow L_b$ and compute the corresponding $u$
                    \STATE $P\leftarrow V_i+s u_i$
                    \IF{$\dist(P,0)=t$ and $\dist(P,C)=t$}
                        \RETURN $\modEJ(A+P)$
                    \ENDIF
                \ENDIF
            \ENDIF
        \ENDFOR
    \ENDFOR
\ENDFOR
\STATE \textbf{return} failure
\end{algorithmic}
\end{algorithm}

\begin{figure}[H]
\centering
\includegraphics[width=0.8\linewidth]{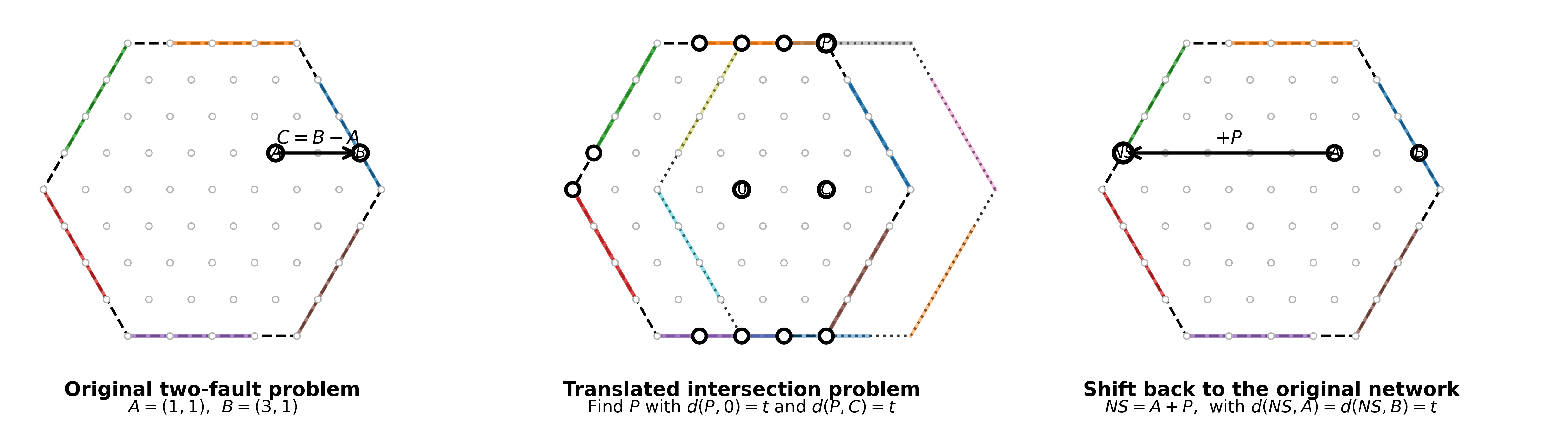}
\caption{Direct new-source selection workflow. The original faults $A$ and $B$ are translated to $0$ and $C=B-A$; a valid translated point $P$ is selected from the boundary intersection; and the final source is obtained by shifting back, $NS=A+P$.}
\label{fig:translation}
\end{figure}

\subsection{Correctness Proofs}
\begin{theorem}[Soundness]
If Algorithm~\ref{alg:direct-ej} returns a node $NS$, then $NS$ is a valid new source for the faulty nodes $A$ and $B$.
\end{theorem}

\begin{proof}
The algorithm returns only after constructing a point $P$ and verifying
\begin{equation}
\dist(P,0)=t,
\qquad
\dist(P,C)=t.
\end{equation}
By the translation correctness lemma, the shifted node $NS=\modEJ(A+P)$ satisfies
\begin{equation}
\dist(NS,A)=t,
\qquad
\dist(NS,B)=t.
\end{equation}
Therefore the returned node is valid.
\end{proof}

\begin{theorem}[Completeness under the re-rooting existence condition]
For any two distinct faulty nodes $A,B\in\EJ_t$ for which the re-rooting existence theorem guarantees a common distance-$t$ node, Algorithm~\ref{alg:direct-ej} returns a valid new source.
\end{theorem}

\begin{proof}
Let $C=\modEJ(B-A)$. By the re-rooting existence theorem, there exists a point $P\in\EJ_t$ such that
\begin{equation}
\dist(P,0)=t,
\qquad
\dist(P,C)=t.
\end{equation}
Thus $P\in\B$ and $P\equiv C+Q \pmod{\alpha}$ for some $Q\in\B$. By Lemma~\ref{lem:finite-shifts}, there is a shift $L\in\Kset$ such that $P=C+Q+L$. Since the six-side partition is disjoint, $P$ lies on a unique side $S_i$ of $\B$, and $Q$ lies on a unique side $S_j$ of $\B$. Hence $P\in S_i\cap(C+S_j+L)$ for some triple $(L,i,j)$.

Algorithm~\ref{alg:direct-ej} checks every triple $(L,i,j)$. For the translated side pair containing $P$, equation \eqref{eq:selector-linear-system} has a feasible integer solution. If the two sides are nonparallel, the solution is unique and Cramer's rule recovers it. If the two sides are parallel, the feasible solutions form an integer interval, and the interval computation contains at least one feasible endpoint or interior point. The algorithm selects one feasible point, verifies the two distance equations, and returns $NS=\modEJ(A+P)$. By the soundness theorem, the returned node is valid.
\end{proof}

\begin{theorem}[Constant-time complexity]
Algorithm~\ref{alg:direct-ej} runs in $O(1)$ time with respect to $t$ and $N$ under the fixed-word arithmetic model.
\end{theorem}

\begin{proof}
The algorithm has three fixed loops: one over the seven shifts in $\Kset$ and two over the six boundary sides. Therefore it checks at most $7\cdot6\cdot6=252$ translated side-pair systems. For each translated system, it performs a fixed number of arithmetic operations: determinant computation, solution by Cramer's rule or parallel interval computation, range checking, and final distance verification. None of these operations requires scanning the $N$ nodes or the $6t$ boundary nodes. Hence the total number of coordinate operations is bounded by a constant independent of $t$ and $N$. Under the fixed-word arithmetic model, the running time is $O(1)$.
\end{proof}

\subsection{Worked Selection Examples}
\textbf{Example 1: faults $0$ and $(1,0)$ in $\EJ_3$.}
Here $t=3$ and $C=(1,0)$. One valid translated point is
\begin{equation}
P=(-2,3).
\end{equation}
Indeed,
\begin{equation}
\dist(P,0)=\max\{2,3,1\}=3,
\end{equation}
and
\begin{equation}
P-C=(-3,3),
\qquad
\dist(P,C)=\max\{3,3,0\}=3.
\end{equation}
Therefore $P$ is a valid new source for faults $0$ and $(1,0)$.

\textbf{Example 2: arbitrary faults in $\EJ_3$.}
Let
\begin{equation}
A=(1,0),
\qquad
B=(2,0).
\end{equation}
Then
\begin{equation}
C=B-A=(1,0).
\end{equation}
Using the valid translated point $P=(-2,3)$ from Example 1, the shifted source is
\begin{equation}
NS=A+P=(-1,3).
\end{equation}
Now
\begin{equation}
NS-A=(-2,3),
\qquad
\dist(NS,A)=3,
\end{equation}
and
\begin{equation}
NS-B=(-3,3),
\qquad
\dist(NS,B)=3.
\end{equation}
Thus $NS=(-1,3)$ is a valid new source for the faults $(1,0)$ and $(2,0)$.

\textbf{Example 3: another translated candidate.}
For faults $0$ and $A=(1,1)$ in $\EJ_3$, the count is $9$ by the counting formula. One valid point is $P=(-2,3)$ because
\begin{equation}
\dist(P,0)=3,
\qquad
P-A=(-3,2),
\qquad
\dist(P,A)=3.
\end{equation}
This example shows that the same distance verification used by the algorithm directly confirms a candidate after the translated side-pair system identifies it.

\section{Complexity Comparison}
\label{sec:complexity}

This section compares the proposed direct selector with two search-based alternatives. Let $N=3t^2+3t+1$.

\subsection{Full Node Scan}
A brute-force method tests every node $X\in\EJ_t$ and checks whether
\begin{equation}
\dist(X,A)=t,
\qquad
\dist(X,B)=t.
\end{equation}
Since there are $N=3t^2+3t+1$ nodes, the full scan has complexity
\begin{equation}
O(N)=O(t^2).
\end{equation}

\subsection{Boundary Scan}
A better search method tests only the boundary around one fault. After translating to $0$ and $C$, it checks all points in $\B$. Since $|\B|=6t$, this method has complexity
\begin{equation}
O(t).
\end{equation}

\subsection{Direct Side-Pair Selection}
Algorithm~\ref{alg:direct-ej} checks at most $7\cdot6\cdot6=252$ translated side-pair systems. Therefore, its complexity is
\begin{equation}
O(1).
\end{equation}
Table~\ref{tab:complexity} summarizes the comparison.

\begin{table}[H]
\caption{Complexity comparison of new-source selection methods in $\EJ_t$.}
\label{tab:complexity}
\centering
\begin{tabular}{lll}
\toprule
Method & Candidates or cases & Complexity\\
\midrule
Full node scan & $N=3t^2+3t+1$ & $O(N)$\\
Boundary scan & $6t$ & $O(t)$\\
Proposed selector & $252$ translated side-pair systems & $O(1)$\\
\bottomrule
\end{tabular}
\end{table}

\begin{figure}[H]
\centering
\includegraphics[width=0.6\linewidth]{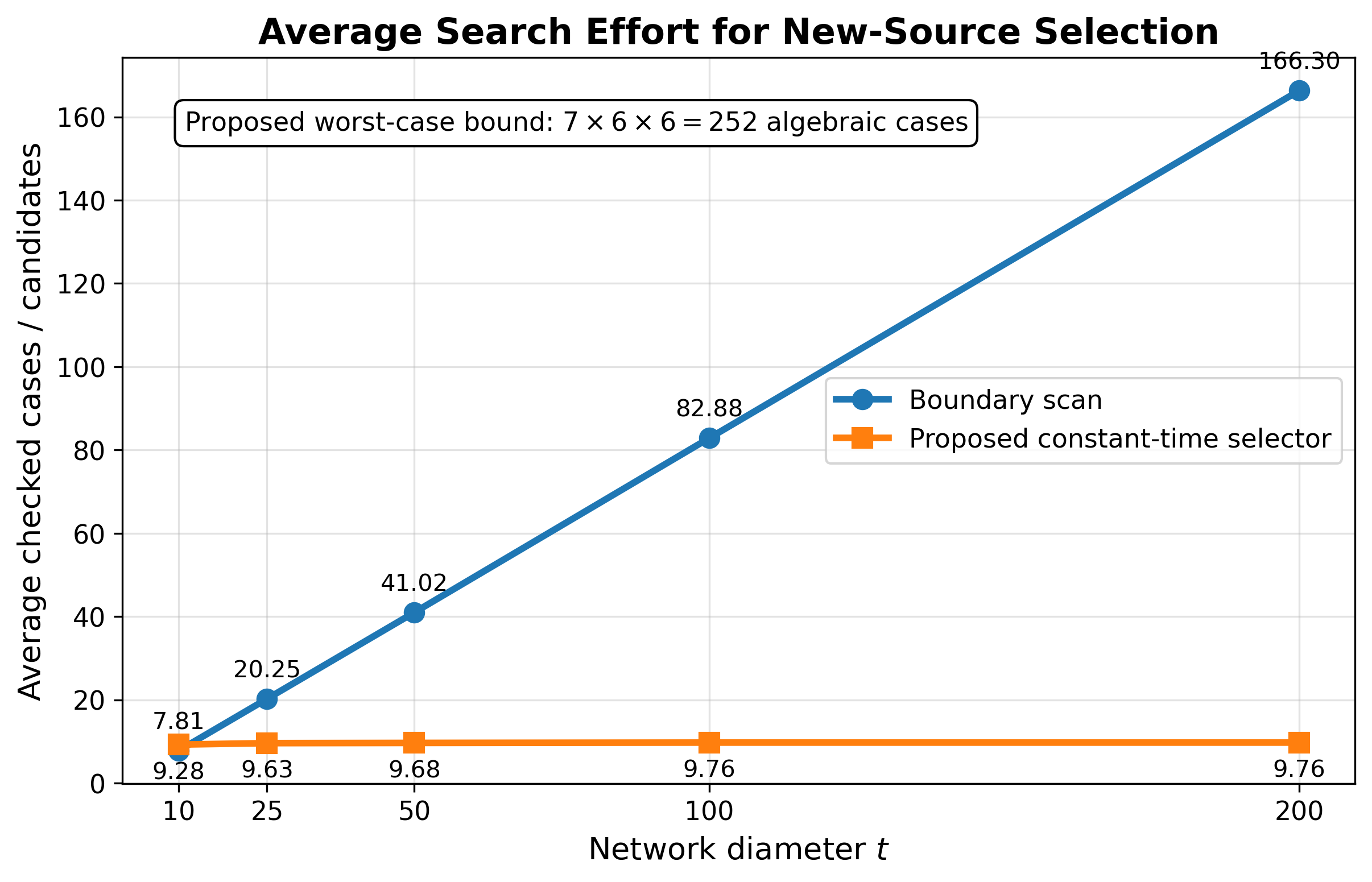}
\caption{Average search effort for new-source selection. Boundary search grows with $t$, while the proposed selector remains nearly constant in the tested cases and is bounded by $252$ translated side-pair systems.}
\label{fig:complexity}
\end{figure}

\section{Computational Validation}
\label{sec:validation}

This section validates the translated side-pair counting and selection algorithms and evaluates the re-rooting recovery performance. The experiments used
\begin{equation}
 t\in\{10,25,50,100,200\},
\end{equation}
with $N=3t^2+3t+1$. Count validation was exhaustive for $t=10,25,50$ and sampled $10{,}000$ nodes for $t=100,200$. Direct-selection validation used $100{,}000$ randomly sampled fault pairs for each value of $t$. The re-rooting experiment used four fault-placement modes, namely random, near, critical, and closepair, with $1000$ trials per exact setting.

\subsection{Experimental Protocol}
All experiments used the canonical coordinate representation defined in Section~\ref{sec:preliminaries}. Node labels were used only for storage and reporting; every distance check was performed by reducing coordinate differences and evaluating the hexagonal distance in \eqref{eq:ej-distance}. For the count-validation experiment, the output of Algorithm~\ref{alg:count-ej} was compared with a brute-force scan of the boundary set $\B$.

For the direct-selection experiment, each sampled fault pair $(A,B)$ was translated to $C=B-A$, Algorithm~\ref{alg:direct-ej} returned a candidate $NS$, and the implementation verified both distance equations in the original coordinates. For the re-rooting experiment, four fault-placement modes were used. The random mode chooses faults uniformly from the non-source nodes. The near mode emphasizes faults close to the original source. The critical mode chooses faults on or near the principal broadcast axes. The closepair mode chooses two faults that are adjacent or nearly adjacent.

\subsection{Counting Validation}
For each tested node $A$, Algorithm~\ref{alg:count-ej} was compared with a brute-force boundary count of all points $P\in\B$ satisfying $\dist(P,A)=t$. Table~\ref{tab:count-validation} shows that the translated side-pair count matched brute force in every tested case.

\begin{table}[H]
\caption{Validation of the translated side-pair count for faults $0$ and $A$.}
\label{tab:count-validation}
\centering
\begin{tabular}{ccccc}
\toprule
$t$ & $N$ & Nodes tested & Mismatches & Max error\\
\midrule
10 & 331 & 331 & 0 & 0\\
25 & 1951 & 1951 & 0 & 0\\
50 & 7651 & 7651 & 0 & 0\\
100 & 30301 & 10000 & 0 & 0\\
200 & 120601 & 10000 & 0 & 0\\
\bottomrule
\end{tabular}
\end{table}

\begin{figure}[H]
\centering
\includegraphics[width=0.6\linewidth]{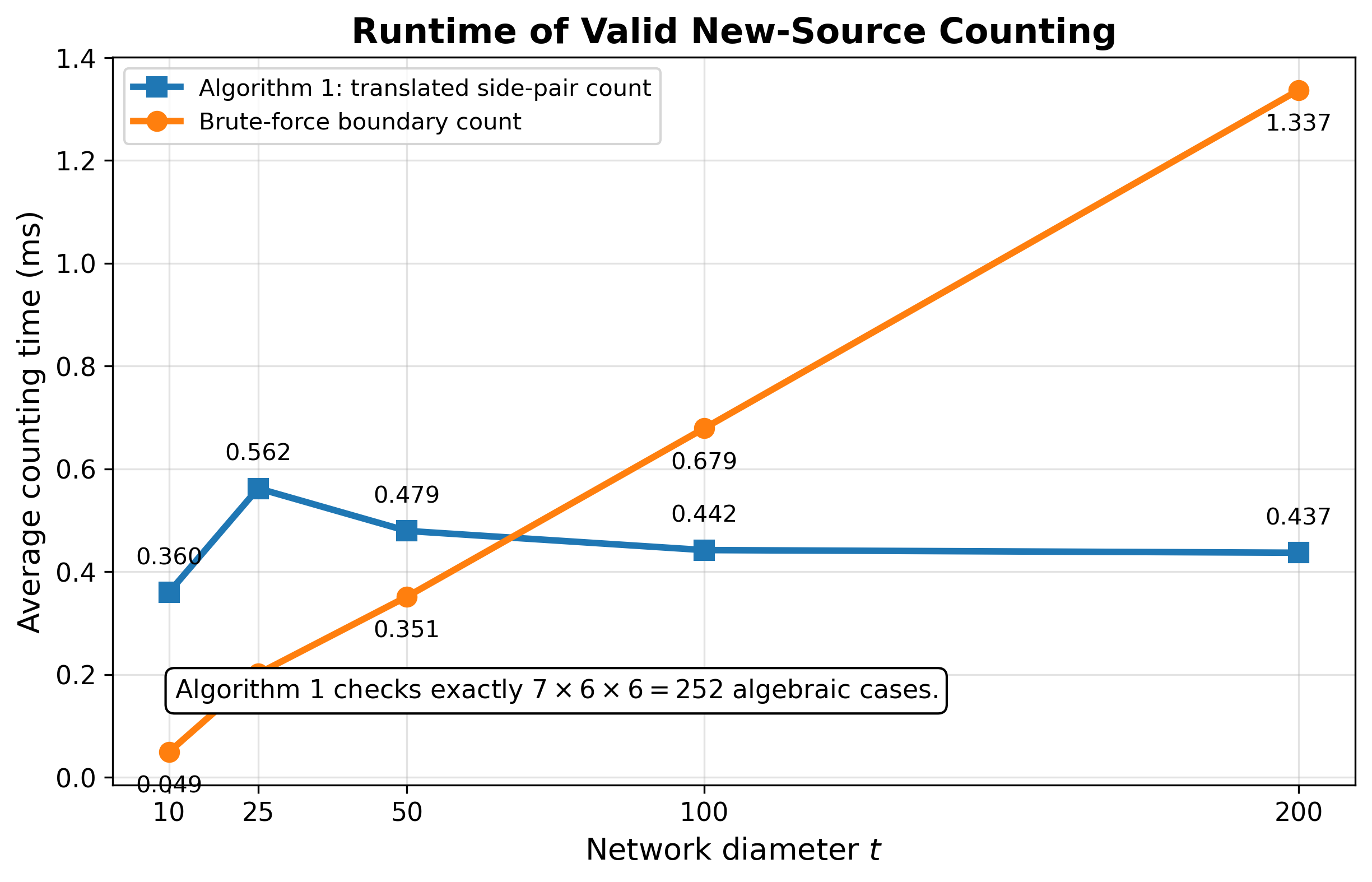}
\caption{Counting-runtime comparison. Algorithm~\ref{alg:count-ej} checks exactly $252$ translated side-pair systems, while brute-force counting scans the $6t$ boundary nodes.}
\label{fig:count-runtime}
\end{figure}

\subsection{Direct-Selection Validation}
For each tested pair $A,B$, Algorithm~\ref{alg:direct-ej} returned a candidate $NS$, and the implementation verified
\begin{equation}
\dist(NS,A)=t,
\qquad
\dist(NS,B)=t.
\end{equation}
Table~\ref{tab:selector-validation} shows that all $500{,}000$ tested fault pairs produced valid outputs.

\begin{table}[H]
\caption{Validation of direct new-source selection over sampled fault pairs.}
\label{tab:selector-validation}
\centering
\begin{tabular}{cccccc}
\toprule
$t$ & $N$ & Pairs & Valid & Failed & Max checked\\
\midrule
10 & 331 & 100000 & 100000 & 0 & 21\\
25 & 1951 & 100000 & 100000 & 0 & 21\\
50 & 7651 & 100000 & 100000 & 0 & 21\\
100 & 30301 & 100000 & 100000 & 0 & 21\\
200 & 120601 & 100000 & 100000 & 0 & 21\\
\bottomrule
\end{tabular}
\end{table}

\subsection{Re-Rooting Recovery Results}
The re-rooting experiment evaluated the complete recovery behavior using one and two faulty nodes. For each value of $t$, each fault count, and each fault-placement mode, $1000$ trials were executed. Thus, the raw experiment contained
\begin{equation}
5\times 2\times 4\times 1000=40000
\end{equation}
trials. Table~\ref{tab:rerooting-summary} aggregates the results over the four modes for each $t$ and fault count. The proposed method achieved $100\%$ success in every setting and reached exactly $N-f$ nodes, where $f$ is the number of faulty nodes.

\begin{table}[H]
\caption{Aggregated re-rooting recovery results over all four fault-placement modes.}
\label{tab:rerooting-summary}
\centering
\begin{adjustbox}{max width=\textwidth}
\begin{tabular}{ccccccccc}
\toprule
$t$ & $N$ & Faults & Trials & Baseline success & Proposed success & Avg. baseline reach & Avg. proposed reach & Expected reach\\
\midrule
10 & 331 & 1 & 4000 & 8.075\% & 100\% & 305.779 & 330.000 & 330\\
10 & 331 & 2 & 4000 & 5.225\% & 100\% & 282.991 & 329.000 & 329\\
25 & 1951 & 1 & 4000 & 4.125\% & 100\% & 1849.122 & 1950.000 & 1950\\
25 & 1951 & 2 & 4000 & 1.550\% & 100\% & 1757.920 & 1949.000 & 1949\\
50 & 7651 & 1 & 4000 & 2.250\% & 100\% & 7320.131 & 7650.000 & 7650\\
50 & 7651 & 2 & 4000 & 1.150\% & 100\% & 7032.936 & 7649.000 & 7649\\
100 & 30301 & 1 & 4000 & 1.375\% & 100\% & 29157.713 & 30300.000 & 30300\\
100 & 30301 & 2 & 4000 & 0.500\% & 100\% & 28182.787 & 30299.000 & 30299\\
200 & 120601 & 1 & 4000 & 0.650\% & 100\% & 116405.605 & 120600.000 & 120600\\
200 & 120601 & 2 & 4000 & 0.250\% & 100\% & 112678.853 & 120599.000 & 120599\\
\bottomrule
\end{tabular}
\end{adjustbox}
\end{table}

\begin{figure}[H]
\centering
\includegraphics[width=0.6\linewidth]{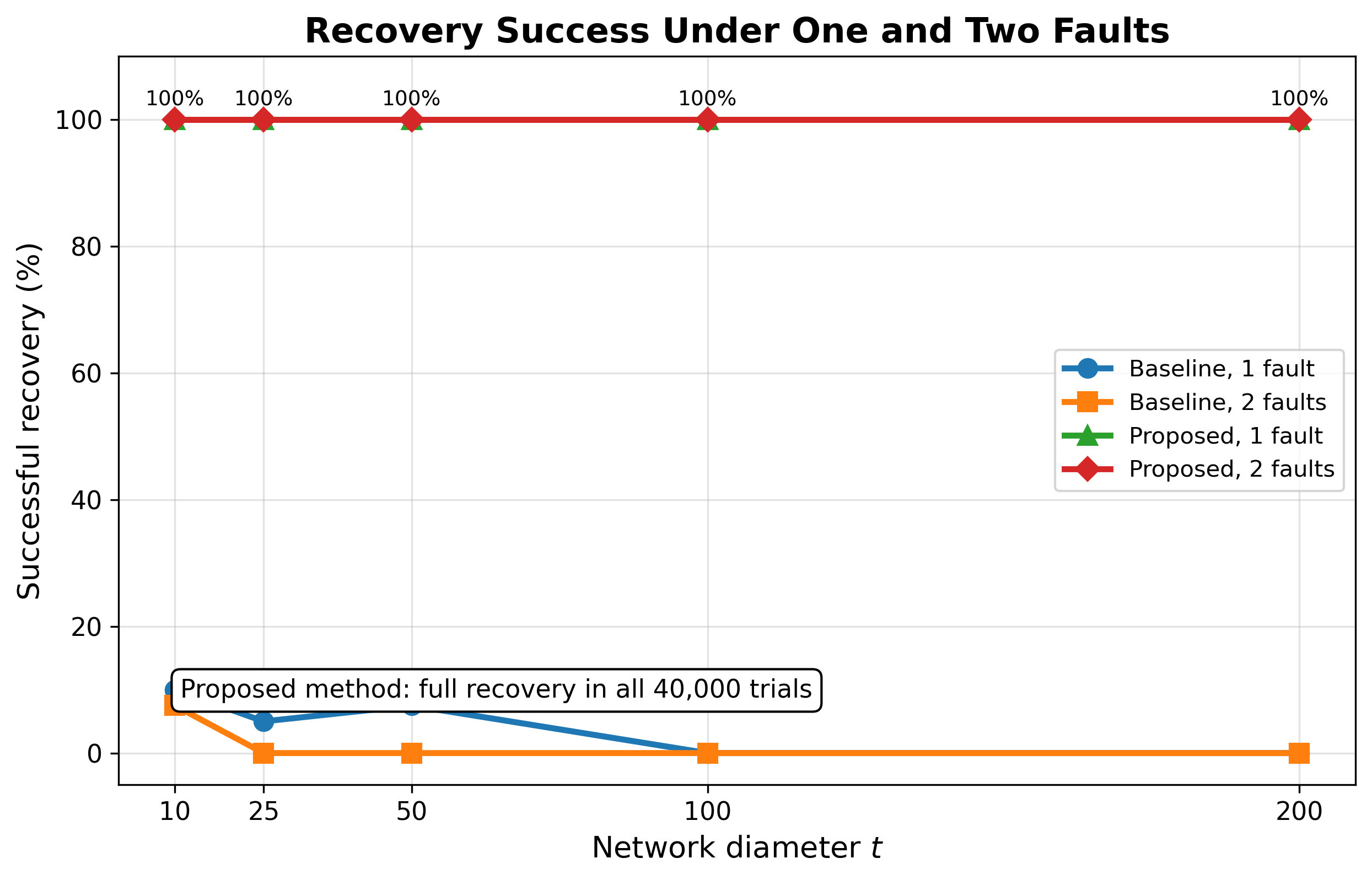}
\caption{Recovery success under one and two faults. The proposed method achieves full recovery in all $40{,}000$ trials, while the baseline broadcast tree loses coverage when faulty nodes interrupt internal broadcast paths.}
\label{fig:recovery-success}
\end{figure}

\begin{figure}[H]
\centering
\includegraphics[width=0.6\linewidth]{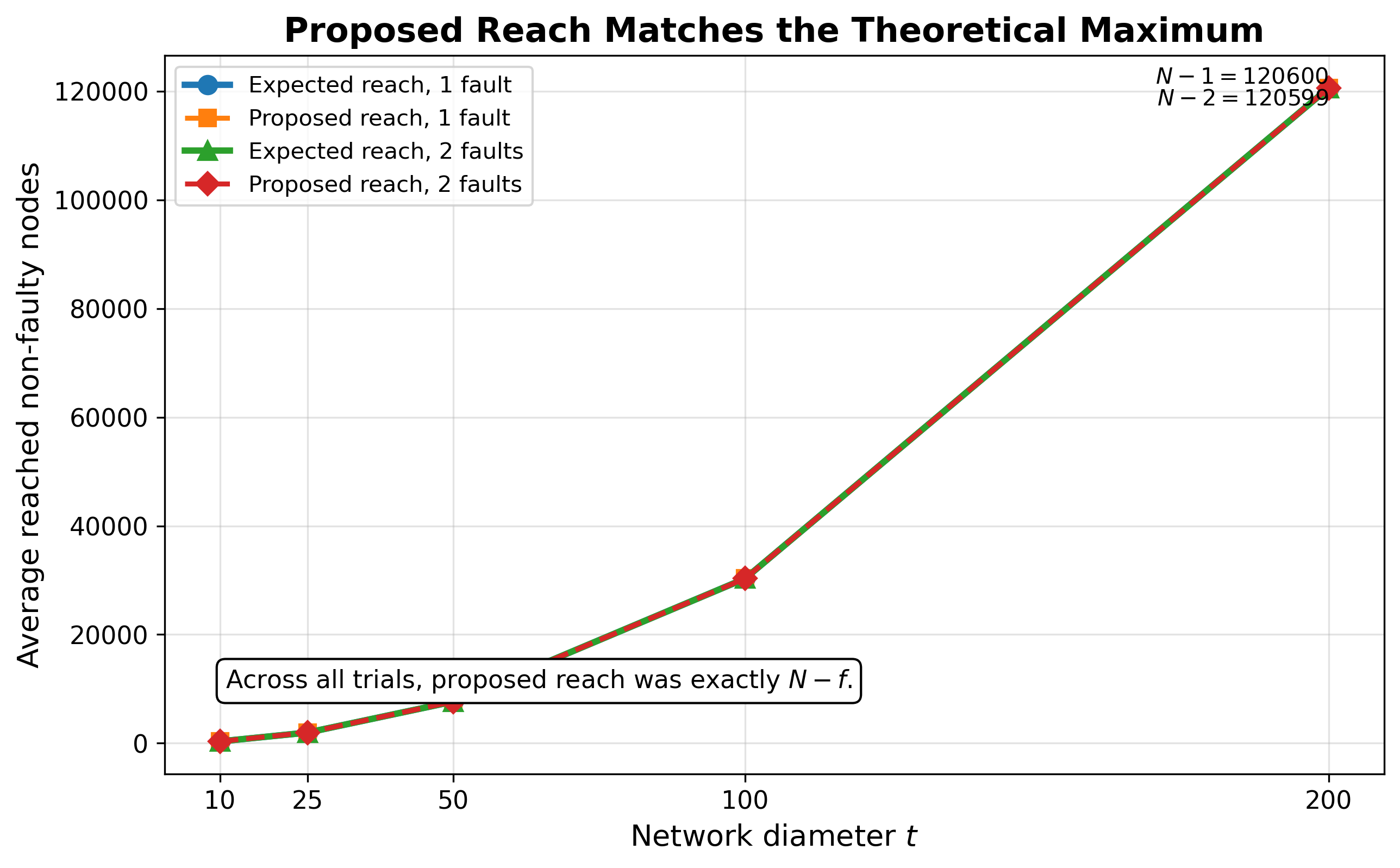}
\caption{Proposed reach versus the theoretical maximum. In all tested settings, the proposed re-rooting method reaches exactly $N-f$ non-faulty nodes, where $f$ is the number of faults.}
\label{fig:proposed-reach}
\end{figure}

\subsection{Runtime Comparison}
The runtime comparison used the same sampled fault pairs for the boundary-search selector and the proposed direct selector. Table~\ref{tab:runtime} and Figure~\ref{fig:runtime} show that the direct method has fixed overhead for small $t$, but it becomes faster as the boundary size grows. At $t=200$, the measured speedup was approximately $9.80\times$.

\begin{table}[H]
\caption{Average runtime per new-source selection query.}
\label{tab:runtime}
\centering
\begin{tabular}{ccccc}
\toprule
$t$ & $N$ & Boundary ms & Direct ms & Speedup\\
\midrule
10 & 331 & 0.008189 & 0.016373 & 0.50$\times$\\
25 & 1951 & 0.019377 & 0.017221 & 1.13$\times$\\
50 & 7651 & 0.046039 & 0.020424 & 2.25$\times$\\
100 & 30301 & 0.095780 & 0.020499 & 4.67$\times$\\
200 & 120601 & 0.214375 & 0.021875 & 9.80$\times$\\
\bottomrule
\end{tabular}
\end{table}

\begin{table}[H]
\caption{Average number of checked candidates or algebraic cases.}
\label{tab:checked-cases}
\centering
\begin{tabular}{cccc}
\toprule
$t$ & Boundary search avg. & Direct avg. & Direct max\\
\midrule
10 & 7.810 & 9.276 & 21\\
25 & 20.247 & 9.630 & 21\\
50 & 41.022 & 9.685 & 21\\
100 & 82.881 & 9.759 & 21\\
200 & 166.298 & 9.757 & 21\\
\bottomrule
\end{tabular}
\end{table}

\begin{figure}[H]
\centering
\includegraphics[width=0.6\linewidth]{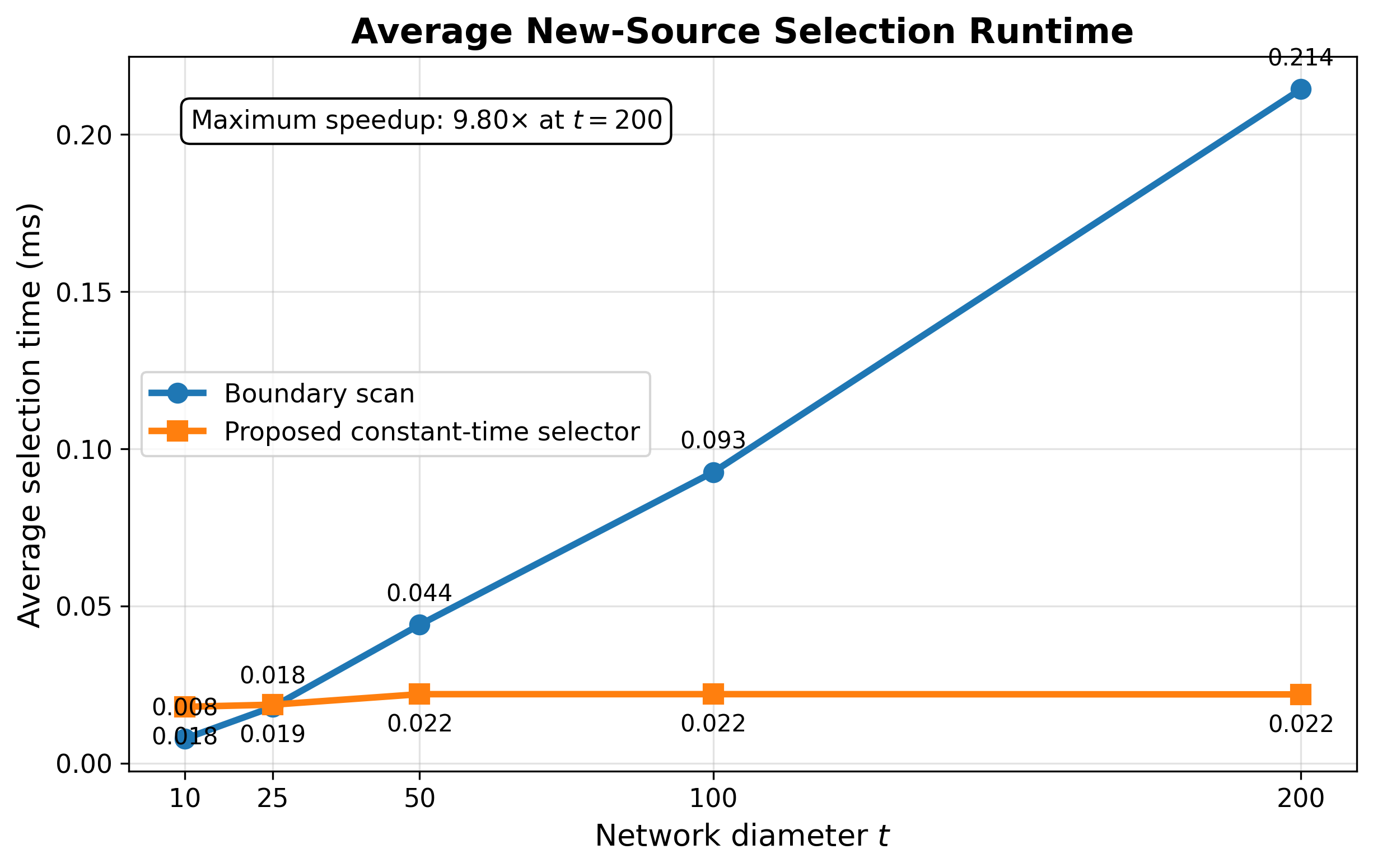}
\caption{Runtime comparison between boundary search and direct selection. The proposed method has fixed algebraic overhead for small $t$, but becomes faster as the boundary size increases.}
\label{fig:runtime}
\end{figure}

\section{Discussion}
\label{sec:discussion}

The contribution of this paper is algorithmic and algebraic rather than a replacement of the underlying broadcasting model. The broadcasting and recovery framework is inherited from the published Eisenstein--Jacobi re-rooting paper, and the present paper strengthens that framework by replacing the new-source selection step with a direct algebraic construction.

The distinction is important. The re-rooting paper answers the existence and recovery question: a suitable new source exists for two faulty nodes and can be used to restore broadcasting. This paper answers the selection question: once two faulty nodes are known, a valid new source can be computed by checking a fixed number of translated side-pair systems. Thus, the two papers are complementary. The first provides the fault-tolerant broadcasting framework, while the present extension improves the algorithmic efficiency of one of its central operations.

The proposed method also explains the geometric reason behind the constant-time bound. The distance-$t$ boundary has $6t$ nodes, but it has only six sides. Searching boundary nodes depends on $t$. Searching translated side-pair systems depends only on the number of sides and the seven relevant quotient shifts, both of which are fixed. Therefore the method avoids growth with network size.

The constant-time method has a fixed overhead. For small values of $t$, a boundary scan may sometimes be competitive because the boundary itself is small. This does not contradict the complexity result. The benefit of the proposed method becomes clearer as $t$ grows: the boundary-search effort increases with $6t$, while the proposed selector remains bounded by the same finite family of translated side-pair systems.

A second point is that the count and the selector serve different purposes. The counting algorithm evaluates all translated side-pair systems and gives the exact number of valid choices for faults $0$ and $A$. The selector stops as soon as it finds one verified candidate. This explains why the theoretical worst-case bound is $252$ systems, while the observed selector checks were much smaller in the sampled experiments.

\section{Conclusion}
\label{sec:conclusion}

This paper developed a closed-form and constant-time method for counting and selecting valid new sources in dense Eisenstein--Jacobi networks. The method extends the published re-rooting-based fault-tolerant broadcasting framework by converting the two-fault source-selection problem into a quotient boundary-intersection problem. For faults $0$ and $A$, the exact number of valid new sources is computed by summing translated side-pair intersections over the seven relevant lattice shifts and the thirty-six hexagonal side pairs. For arbitrary faults $A$ and $B$, the problem is translated to the origin and the difference node $C=B-A$, solved by the same translated side-pair systems, and shifted back to obtain the new source.

The correctness of the method was proved through boundary partition, finite-shift sufficiency, counting-algorithm correctness, selector soundness, selector completeness, and complexity arguments. Since both algorithms evaluate at most $7\cdot6\cdot6=252$ translated side-pair systems and perform only constant-time arithmetic in each case, both run in $O(1)$ time under the fixed-word arithmetic model. Computational validation confirmed that the count matched brute force in all tested cases, the direct selector returned valid outputs for all $500{,}000$ sampled fault pairs, and the re-rooting method achieved $100\%$ recovery over $40{,}000$ trials.

Future work includes ranking multiple valid candidates according to broadcast recovery cost, studying implementation-level integration with the complete re-rooting broadcast procedure, and investigating whether similar fixed-case constructions can be developed for larger fault sets or other algebraic network families.

\appendix
\section{Translated Side-Pair System Details}
\label{app:sidepair-details}

This appendix gives the general algebraic form used by both Algorithm~\ref{alg:count-ej} and Algorithm~\ref{alg:direct-ej}. For a center node $X$, which is either $A$ in the counting problem or $C=B-A$ in the selection problem, a translated side-pair intersection has the form
\begin{equation}
S_i\cap (X+S_j+L),
\qquad L\in\Kset.
\end{equation}
Using the side parameterization $S_i=\{V_i+s u_i:0\leq s\leq t-1\}$, this intersection is described by
\begin{equation}
V_i+s u_i=X+V_j+u u_j+L,
\qquad 0\leq s,u\leq t-1.
\label{eq:app-side-pair}
\end{equation}
Moving all parameter terms to the left gives
\begin{equation}
M_{ij}
\begin{bmatrix}
s\\u
\end{bmatrix}
=
b_{Lij}(X),
\end{equation}
where
\begin{equation}
M_{ij}=[u_i\;-u_j],
\qquad
b_{Lij}(X)=X+V_j+L-V_i.
\end{equation}

If $\det(M_{ij})\ne 0$, then the side directions are nonparallel and the system has at most one solution. Writing
\begin{equation}
u_i=(r_1,r_2),\qquad -u_j=(r_3,r_4),
\qquad b_{Lij}(X)=(b_1,b_2),
\end{equation}
the determinant is
\begin{equation}
\Delta=r_1r_4-r_2r_3.
\end{equation}
The candidate parameters are
\begin{equation}
s=\frac{b_1r_4-r_3b_2}{\Delta},
\qquad
u=\frac{r_1b_2-b_1r_2}{\Delta}.
\end{equation}
This translated side pair contributes one point if and only if $s$ and $u$ are integers and both lie in the interval $[0,t-1]$. Otherwise it contributes zero points.

If $\det(M_{ij})=0$, then the side directions are parallel. In that case, consistency is checked by testing whether $b_{Lij}(X)$ is parallel to $u_i$. When it is not parallel, the contribution is zero. When it is parallel, the two-parameter equation reduces to one integer linear relation between $s$ and $u$. The parameter bounds then determine an interval of feasible integer values. The count algorithm adds the length of that interval, while the selector chooses one feasible parameter value and verifies the resulting point.

This appendix also clarifies why the algorithms do not scan the boundary. A boundary scan enumerates all $6t$ values of $P$. The translated side-pair method instead evaluates a fixed collection of systems indexed by
\begin{equation}
(L,i,j)\in \Kset\times\{1,\ldots,6\}\times\{1,\ldots,6\}.
\end{equation}
The number of such triples is $7\cdot6\cdot6=252$, independent of $t$.

\section{Parallel Side-Pair Computation}
\label{app:parallel}

This appendix gives a compact implementation form for the parallel case used in Algorithm~\ref{alg:direct-ej}. Consider the side-pair equation
\begin{equation}
V_i+s u_i=C+V_j+u u_j+L,
\end{equation}
with $0\leq s,u\leq t-1$. If $\det([u_i\;-u_j])=0$, then the two direction vectors are parallel. Consistency is checked by verifying that $C+V_j+L-V_i$ is parallel to $u_i$. If it is not parallel, the contribution is zero.

If the system is consistent, one coordinate equation determines a relation of the form
\begin{equation}
u = \rho s + \sigma,
\end{equation}
where $\rho\in\{1,-1\}$ for the side directions used in this paper and $\sigma$ is an integer offset. The range restrictions become
\begin{equation}
0\leq s\leq t-1,
\qquad
0\leq \rho s+\sigma\leq t-1.
\end{equation}
These inequalities give an interval
\begin{equation}
L\leq s\leq U.
\end{equation}
The number of feasible integer solutions is $\max\{0,U-L+1\}$. For selection, Algorithm~\ref{alg:direct-ej} chooses $s=L$ when $L\leq U$, computes $u=\rho L+\sigma$, and verifies the two distance equations before returning.

\section*{Acknowledgments}
The authors would like to acknowledge the support of Kuwait University and its Computer Science Department.

\end{document}